\documentclass[11pt]{article}

\usepackage[margin=1in]{geometry}
\usepackage{setspace}
\usepackage[T1]{fontenc}
\usepackage[utf8]{inputenc}
\usepackage{lmodern}
\usepackage{microtype}

\usepackage{amsmath,amssymb,amsthm,mathtools,bm}

\usepackage{graphicx}
\usepackage{tikz}
\usetikzlibrary{patterns}
\usetikzlibrary{decorations.pathreplacing}
\usepackage{booktabs}
\usepackage{caption}
\usepackage{subcaption}
\usepackage{float}
\usepackage{rotating} 

\usepackage{enumitem}
\usepackage{xcolor}

\usepackage[round,authoryear]{natbib}
\usepackage{hyperref}
\hypersetup{
  colorlinks=true,
  linkcolor=blue,
  citecolor=blue,
  urlcolor=blue
}
\usepackage{cleveref} 

\numberwithin{equation}{section}
\newtheorem{theorem}{Theorem}[section]
\newtheorem{proposition}[theorem]{Proposition}

\theoremstyle{definition}

\theoremstyle{remark}

\newcommand{\R}{\mathbb{R}}

\title{The Welfare Effects of Policy Signalling in a Regime Change Game}

\author{Georgy Lukyanov\thanks{Toulouse School of Economics, 1 Esp.\ de l'Universit\'e, Toulouse, 31000, France. Email: Georgy.Lukyanov@tse-fr.eu}}

\date{}

\begin{document}

\maketitle

\abstract{We explore a specific parametric example of a regime change game in which a policymaker defends the status quo against a continuum of atomistic agents who seek to overthrow it. At the start of the game, the policymaker can initiate a policy intervention to make an attack less appealing. When agents’ private information is relatively imprecise, signalling is beneficial to the policymaker when the fundamentals are strong but detrimental when the fundamentals are weak.}

\textbf{Keywords:} Global games, Signalling, Policymaker, Regime change

\textbf{JEL classification:} D82, D84, E58, F31

\medskip

\noindent
\begingroup
\setlength{\fboxsep}{8pt}\setlength{\fboxrule}{0.4pt}%
\fbox{%
  \begin{minipage}{0.85\textwidth}
  \small
  \textbf{Author's accepted manuscript (postprint)} of:\\
  Lukyanov, G. (2025). \emph{The welfare effects of policy signalling in a regime change game}. 
  \textit{Economic Theory Bulletin}. \textbf{Published online:} 17 June 2025.\\
  \textbf{Version of Record (VoR) DOI:}
  \href{https://doi.org/10.1007/s40505-025-00295-z}{10.1007/s40505-025-00295-z}.\\[2pt]
  \emph{This is a post–peer-review, pre-copyedit version of an article published in \textit{Economic Theory Bulletin}.
  The final authenticated version is available at the DOI above.}
  \end{minipage}%
}
\endgroup

\medskip

\section{Introduction}

Coordinated actions by large groups often drive major macroeconomic events, such as currency attacks, bank runs, stock market crashes, and political regime changes. In these scenarios, many individuals may disrupt the status quo—by selling short a pegged currency, triggering a bank run, or joining mass protests—which can force a central bank to devalue its currency, plunge a commercial bank into a liquidity crisis, or even topple a political leader.

Typically, these events feature numerous small players acting against a large entity—the policymaker—whose goal is to preserve the status quo. In this paper, we consider a sequential game in which the policymaker can, before a coordinated attack unfolds, undertake a costly action to deter such disruptions. For example, in anticipation of a currency attack, a central bank might raise interest rates to increase speculative borrowing costs; similarly, a government might bolster its security forces and tighten laws to deter protests.

An interesting twist arises when the policymaker possesses additional information about the regime’s resilience. Preventive actions may inadvertently signal this information to the public. For instance, a heavy police presence might lead citizens to infer that the political leader fears the regime’s potential collapse, provoking even more aggressive behavior.

In our model, all payoffs depend on the fundamental state $\theta$, which measures the regime’s strength. We focus on a signaling equilibrium in which the policymaker’s intervention $r$ influences the agents’ beliefs about $\theta$.  

Our main result (Proposition \ref{prop:central}) characterizes how the effect of raising $r$ on the policymaker’s welfare depends on the precision of the agents’ signals.  We show that: if the private signals are very precise (low noise), then higher $r$ uniformly reduces the policymaker’s welfare for *all* values of the fundamental.  By contrast, if the private signals are very noisy, then increasing $r$ hurts the policymaker when the fundamental $\theta$ is low (weak regime) but helps the policymaker when $\theta$ is high (strong regime).  Intuitively, with noisy signals a large $r$ can reassure agents under strong fundamentals and deter attacks, whereas with precise signals it only adds cost without improving beliefs. 

The rest of the paper is organized as follows. Section \ref{sec:model} describes the model.  Section \ref{sec:characterisation} provides the equilibrium characterization and compares welfare under different information precisions.  A discussion of related literature is deferred to Section \ref{sec:literature}, and Section \ref{sec:conclusion} concludes.

\section{\label{sec:model}The Model}

\subsection{Players.} The game is played by two types of players: a single \emph{policymaker} and a unit continuum (mass 1) of identical \emph{agents}. The policymaker’s action is a policy choice $r$, and each agent’s action is $a_i \in \{0,1\}$ (for example, $a_i=1$ may denote “attack” and $a_i=0$ “stay put”).

\subsection{Sequence of play.} The extensive-form timing is as follows:
\begin{enumerate}
  \item \textbf{Nature draws $\boldsymbol{\theta}$.} A fundamental state $\theta \in \mathbb{R}$ is drawn by Nature. We assume $\theta$ is distributed according to a uniform distribution over $\mathbb{R}$ (treated as an improper prior with respect to Lebesgue measure on $\mathbb{R}$).
  \item \textbf{Policymaker’s signal (policy choice).} The policymaker observes the realized $\theta$ and then chooses a policy $r \in \mathbb{R}_+$.
  \item \textbf{Agents’ information and actions.} Each agent $i$ observes the announced policy $r$ and receives a private signal $x_i \in \mathbb{R}$ about $\theta$ (but does not observe $\theta$ directly). Based on $r$ and $x_i$, the agent chooses an action $a_i \in \{0,1\}$.
  \item \textbf{Policymaker’s final move.} After observing the agents’ actions (e.g., the fraction of agents who chose $a_i=1$), the policymaker chooses a final decision $d \in \{0,1\}$.
\end{enumerate}

\subsection{Payoffs.} Each agent $i$ receives a payoff that depends on her action $a_i$, the fundamental $\theta$, and the policymaker’s final decision $d$. The interpretation is as follows: $a_i = 1$ denotes an “attack,” $d = 1$ denotes “regime change” (e.g., policy abandonment), and $\theta$ is the true strength of the regime. Specifically, the payoff to agent $i$ is: 
\[
u_i(a_i,\theta,d) = 
\begin{cases}
-r & \text{if } a_i = 1 \text{ and } d = 0,\\
0 & \text{if } a_i = 0,\\
1 - r & \text{if } a_i = 1 \text{ and } d = 1~,
\end{cases}
\] 
meaning that if the agent attacks and the regime survives ($d=0$), she suffers a loss. If she abstains ($a_i=0$), her payoff is normalized to zero. If she attacks and the regime falls ($d=1$), she earns a normalized success payoff of 1.

The policymaker observes the aggregate behavior (e.g., the fraction of agents choosing $a_i=1$) and chooses $d \in \{0,1\}$. Let $\alpha \in [0,1]$ denote the proportion of agents who attack. The policymaker’s payoff is: 
\[
U(r,d,\theta,\alpha) = (1-d)\big(\theta - \alpha\big) - c(r)\,,
\] 
where $c(r)$ is the cost of implementing policy $r$.  In particular we assume $c(r)$ is convex in $r$ and for concreteness take 
\[
c(r) \;=\; \frac{1}{2}\,(r - \underline{r})^2,
\] 

The policymaker’s decision $d$ depends on $\theta$ and the attack level $\alpha$ through her payoff $U(r,d,\theta,\alpha) = (1-d)(\theta-\alpha) - c(r)$.  In particular, higher $\theta$ or lower $\alpha$ make keeping the regime ($d=0$) more attractive.  Thus, we will see below that the equilibrium threshold for $\theta$ increases with $r$ and that a larger fraction of attackers $\alpha$ tends to trigger regime change ($d=1$).  

We analyze the game’s Perfect Bayesian Equilibrium.

\section{\label{sec:characterisation}Equilibrium Characterization}

Assume that the policymaker and the agents hold an improper uniform prior on the fundamental $\theta$ (i.e. $\theta \sim U(\mathbb{R})$). For analytical tractability, we maintain the quadratic cost specification $c(r) = \frac{1}{2}(r-\underline{r})^2$ introduced above.

We begin by describing the equilibrium in the speculators' game for the case where $r$ is set exogenously and hence is \emph{not} treated as a signal by the agents.

\subsection{Continuation game equilibrium for a given $r$}

Consider the continuation game with a given $r$. Each agent $i$ observes $r$ and receives a private signal 
$x_i = \theta + \epsilon_i$, where $\epsilon_i$ is uniformly distributed on $[-\sigma,\sigma]$ (so $\sigma>0$ measures the noise level).

Agents set their strategies 
\[
a: \R \times (0,1) \to \{0,1\}\,,
\] 
where $a(x_i,r)$ is the agent's rule for attacking (1) or not (0) given her private signal $x_i$ and the publicly observed policy $r$.

Since $x$ is positively correlated with $\theta$, a lower $x$ indicates weaker fundamentals and a higher chance that a coordinated attack will topple the regime. An agent’s dominant strategy is to attack if her signal $x_i$ is below some cutoff $\tilde{x}$; that is,
\[
a^*(x_i) = 
\begin{cases}
1, & \text{if } x_i \le \tilde{x},\\[1mm]
0, & \text{if } x_i > \tilde{x}\,.
\end{cases}
\] 
Because the policymaker will abandon the regime when the aggregate attack exceeds the fundamental, the probability of a successful attack is 
\[
A(\theta) = \Pr\{x_i \le \tilde{x} \mid \theta\}\,.
\] 
Because $A(\theta)$ is continuous, bounded, and decreasing in $\theta$, there is a unique fixed point $\tilde{\theta}$ of $A$, i.e. a $\tilde{\theta}$ satisfying $A(\tilde{\theta}) = \tilde{\theta}$.

This reasoning leads to an equilibrium for a given $r$, defined by the thresholds $(\tilde{x}(r), \tilde{\theta}(r))$. Proposition \ref{prop:thresholds} states:

\begin{proposition}\label{prop:thresholds}
For each $r \in \mathbb{R}_+$, there is a unique equilibrium in the continuation game characterized by thresholds $(\tilde{x}(r), \tilde{\theta}(r))$.

In equilibrium, each agent attacks if and only if her signal $x_i \le \tilde{x}(r)$, and the regime falls if and only if the fundamental $\theta \le \tilde{\theta}(r)$.  Solving the iterative threshold conditions yields explicit formulas.  In particular, we find 
\[
\tilde{x}(r) \;=\; (1+2\sigma)\,(1-r) - \sigma, 
\qquad 
\tilde{\theta}(r) \;=\; 1 - r,
\]
so that in equilibrium $\tilde{\theta}(r)=1-r$ as stated in the Proposition.
\end{proposition}

\begin{proof}
This is a straightforward adaptation of standard results in the global-games literature. See, for example, Theorem 1 of \citet{MorrisShin1998} or the proof of Proposition A.1 in \citet{Edmond2013}.
\end{proof}

In the unique equilibrium that survives iterated elimination of conditionally dominated strategies, each agent attacks if and only if her signal is below $\tilde{x}(r)$, and the regime collapses if $\theta \le 1 - r$. This result follows from an iterative conditional dominance argument, which constructs two sequences of cutoff signals, $\{\tilde{x}_k\}$ and $\{\tilde{x}^k\}$, for the agents with the most optimistic and most pessimistic beliefs, and shows that both sequences converge to the same limit $\tilde{x}$. The existence of \emph{dominance regions} allows us to determine $\tilde{x}$ uniquely, thereby ensuring equilibrium uniqueness.

\subsection{Signaling by the Choice of Policy}

Next, consider the scenario in which agents anticipate that the policymaker’s choice of $r$ may depend on the realized $\theta$. In this case, speculators can infer $\theta$ from the observed policy $\bar{r}(\theta)$ and coordinate accordingly. Here, $\theta$ (observed by the policymaker) serves as the policymaker’s type, which is signaled to the agents through the choice of $r$.

Following \citet{Angeletosetal2006}, we construct an \emph{active policy equilibrium} in which the policymaker intervenes only for an intermediate range of fundamentals, adapting their Proposition 2 construction to our parametric setting.

Since active policy intervention is costly and is undertaken only if the policymaker intends to maintain the status quo, a higher $r$ signals that fundamentals exceed some lower threshold $\underline{\theta}$. Such common knowledge that the economy is not too weak thereby eliminates the lower dominance region where agents would otherwise prefer to attack.

Consequently, any policymaker who raises the policy to a level exceeding $\underline{r}$ completely removes the incentive to attack. In this parametric example, for any $r' \in (\underline{r}, \tilde{r}]$ (with $\tilde{r}$ satisfying $c(\tilde{r}) = \tilde{\theta}$, i.e.\ $\tfrac{1}{2}(\tilde{r}-\underline{r})^2 = 1 - \underline{r}$), there exist thresholds $\underline{\theta}(r')$ and $\overline{\theta}(r')$, along with a signal threshold $x'(r')$, given by:
\begin{flalign}
\label{eq:fortlow}
\underline{\theta}(r') &= \frac{1}{2}\big(r' - \underline{r}\big)^2,\\
\overline{\theta}(r') &= 2\sigma + \Big[1 - \frac{2\sigma}{\,1-\underline{r}\,}\Big]\underline{\theta}(r'),\\
x'(r') &= \overline{\theta}(r') + \sigma\big(2\,\underline{\theta}(r') - 1\big)\,. 
\end{flalign}
Here, $\underline{\theta}(r')$ is the minimum fundamental level that justifies intervention, $\overline{\theta}(r')$ is the level at which the policymaker is indifferent between intervening and not, and $x'(r')$ is the signal threshold triggering an attack.

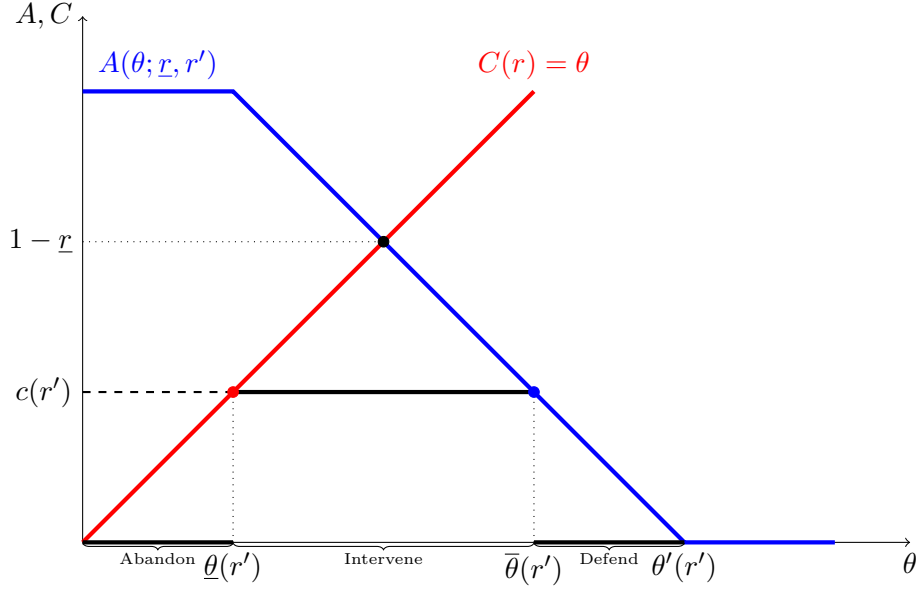
\begin{figure}
\begin{center}
\begin{tikzpicture}[scale=1]
\draw[<->] (0,7) node[left]{{$A,C$}} -- (0,0) -- (11,0) node[below]{{$\theta$}};
\draw[dotted] (0,4) node[left]{{$1-\underline{r}$}} -- (4,4);
\draw[ultra thick, blue] (0,6) -- node[above]{{$A(\theta;\underline{r},r')$}} (2,6) -- (8,0) -- (10,0);
\draw[ultra thick, red] (0,0) -- (6,6) node[above]{{$C(r) = \theta$}};
\draw[thick, dashed] (0,2) node[left]{{$c(r')$}} -- (6,2);
\draw[dotted] (2,0) node[below]{{$\underline{\theta}(r')$}} -- (2,2);
\draw[dotted] (6,0) node[below]{{$\overline{\theta}(r')$}} -- (6,2);
\node[below] at (8,0) {{$\theta'(r')$}};
\draw[ultra thick] (0,0) -- (2,0);
\draw[ultra thick] (2,2) -- (6,2);
\draw[ultra thick] (6,0) -- (8,0);
\draw [decorate, decoration={brace}] (6,-.01) -- node[below] {\tiny{Intervene}} (2,-.01);
\draw [decorate, decoration={brace}] (2,-.01) -- node[below] {\tiny{Abandon}} (0,-.01);
\draw [decorate, decoration={brace}] (8,-.01) -- node[below] {\tiny{Defend}} (6,-.01);
\draw[fill] (4,4) circle[radius=2pt];
\draw[fill,red] (2,2) circle[radius=2pt];
\draw[fill,blue] (6,2) circle[radius=2pt];
\end{tikzpicture}
\caption{Aggregate attack and the regions of intervention and non-intervention.}
\label{fig:actpol}
\end{center}
\end{figure}

Figure \ref{fig:actpol} illustrates the equilibrium. For $\theta \in [\,\underline{\theta}(r'),\,\overline{\theta}(r')\,]$, those types of policymakers intervene; outside this interval they do not. In particular, if the policymaker raises the policy to $r'$, agents infer that $\theta > \underline{\theta}(r')$ and coordinate on not attacking, so that $A=0$ for every $\theta$. The curve $A(\theta;\underline{r},r')$ represents the aggregate attack in the continuation game following \emph{no intervention} ($r=\underline{r}$) in the signaling equilibrium where the intervening policymaker is supposed to play $r = r'$.

The type with $\underline{\theta}(r')$ is indifferent between intervening and not. If that policymaker intervenes (raising the policy to $r' > \underline{r}$), they avoid an attack (receiving a payoff of $\theta$ rather than 0) at the cost $c(r') > 0$; these must balance to yield indifference. For $\theta > \underline{\theta}(r')$, the benefit of intervention increases with $\theta$ while the opportunity cost from the aggregate attack $A(\theta;\underline{r},r')$ decreases. At $\theta = \overline{\theta}(r')$, the policymaker is indifferent between intervening (incurring cost $c(r')$) and not intervening (facing a partial attack $A(\theta;\underline{r},r') > 0$, yielding a net payoff of $\theta - A(\theta;\underline{r},r')$).\footnote{Recall that no attack occurs if the policy is raised to $r'$; the policymaker will always defend the regime since $\overline{\theta} > A(\overline{\theta};\underline{r},r')$.}

Finally, define $\theta'(r')$ as the minimal $\theta$ for which $A(\theta;\underline{r},r') = 0$.

When agents observe no intervention ($r = \underline{r}$), they infer that $\theta$ is either very strong ($\theta > \overline{\theta}(r')$) or very weak ($\theta < \underline{\theta}(r')$). Their attack decision—determined by the threshold $x'(r')$—and the resulting aggregate attack $A(\theta;\underline{r},r')$ follow from these posterior beliefs given $\theta \notin [\,\underline{\theta}(r'),\,\overline{\theta}(r')\,]$.

A continuum of equilibria exists for $r' \in (\underline{r}, \tilde{r}]$, which are generally not Pareto-ranked. Nonetheless, some implications for the policymaker's ex post utility can be derived.

For the equilibrium policy level $r'$ played by types $\theta \in [\,\underline{\theta}(r'),\,\overline{\theta}(r')\,]$, the policymaker's ex post utility is 
\begin{equation}
U(\theta; r') = 
\begin{dcases}
0, & \text{if } \theta < \underline{\theta}(r'),\\
\theta - c(r'), & \text{if } \theta \in [\,\underline{\theta}(r'),\,\overline{\theta}(r')\,),\\
\theta - A(\theta;\underline{r},r'), & \text{if } \theta \ge \overline{\theta}(r')\,,
\end{dcases}
\end{equation}
where we recall that $\tilde{\theta}$ denotes the unique solution to $A(\theta) = \theta$ (as defined above).

Let $\tilde{r} := \underline{r} + \sqrt{2(1 - \underline{r})}$. For each $r' \in (\underline{r},\,\tilde{r}]$, define the policymaker’s strategy $\bar{r}_{r'}: \mathbb{R} \to \{\underline{r},\,r'\}$ by 
\[
  \bar{r}_{r'}(\theta) = 
  \begin{cases}
    r' & \text{if } \theta \in [\,\underline{\theta}(r'),\,\overline{\theta}(r')\,],\\[6pt]
    \underline{r}  & \text{otherwise}\,,
  \end{cases}
\] 
where $\underline{\theta}(r') < \overline{\theta}(r')$ are the threshold values defined above. Then define the corresponding signal cutoff for agents, $x'_{r'}$, such that an agent with private signal $x_i = x'_{r'}$ is indifferent between choosing $a_i=0$ and $a_i=1$ when the policy is $r'$.

We now compare the policymaker's payoff as a function of $r'$ — that is, across different equilibria corresponding to different $r'$ values. We find that when the noise in agents' private signals ($\sigma$) is sufficiently large, a more aggressive intervention (a higher $r'$) hurts a weak policymaker but benefits a strong one. In contrast, when $\sigma$ is small, a higher $r'$ uniformly harms all types of policymakers.

This result is formalized in the following proposition:

\begin{proposition}\label{prop:central}
Define $\theta'(r') := \min\{\theta \in \R:~A(\theta;\underline{r},r') = 0\}$. 
When $\sigma > \frac{\,1-\underline{r}\,}{\,2\,\underline{r}\,}$, the policymaker's ex post welfare $U(\theta; r')$ is \emph{decreasing} in $r'$ for 
\[
\theta \in [\,\underline{\theta}(r'),\,\overline{\theta}(r')\,]
\] 
and \emph{increasing} in $r'$ for 
\[
\theta \in [\,\overline{\theta}(r'),\,\theta'(r')\,]\,.
\] 
The welfare is unaffected by $r'$ for $\theta < \underline{\theta}(r')$ and $\theta > \theta'(r')$.

When $\sigma \le \frac{\,1-\underline{r}\,}{\,2\,\underline{r}\,}$, the policymaker's welfare is uniformly decreasing in $r'$ (for all $\theta$).
\end{proposition}

\begin{proof}
See Appendix~\ref{app:prop2}.
\end{proof}

\noindent\textit{Discussion.} Figure \ref{fig:welfare} illustrates the welfare comparison in Proposition \ref{prop:central}. The solid black curve in Figure \ref{fig:welfare} depicts the policymaker’s payoff $U(\theta;r')$ under a given intervention $r'$ (this piecewise-linear curve reflects four regions: for $\theta < \underline{\theta}(r')$ the regime collapses with payoff 0; for $\underline{\theta}(r') \le \theta < \overline{\theta}(r')$ the policymaker intervenes and earns a positive payoff; for $\overline{\theta}(r') \le \theta < \theta'(r')$ the policymaker does not intervene but faces a partial attack; and for $\theta > \theta'(r')$ no attack occurs, yielding payoff $\theta$). The dashed red curve shows the payoff under a more aggressive policy $r'' > r'$. Consistent with Proposition~\ref{prop:central}, when private signals are \emph{noisy} (large $\sigma$), the red curve lies below the black curve at low $\theta$ (a weak policymaker incurs a higher cost and is worse off) but crosses above it at higher $\theta$ (a strong policymaker deters attacks and is better off with a higher $r$). By contrast, when signals are \emph{precise} (small $\sigma$), the red curve never exceeds the black curve, indicating that a higher $r$ uniformly harms the policymaker. This confirms that Figure~\ref{fig:welfare} and Proposition~\ref{prop:central} are fully consistent.

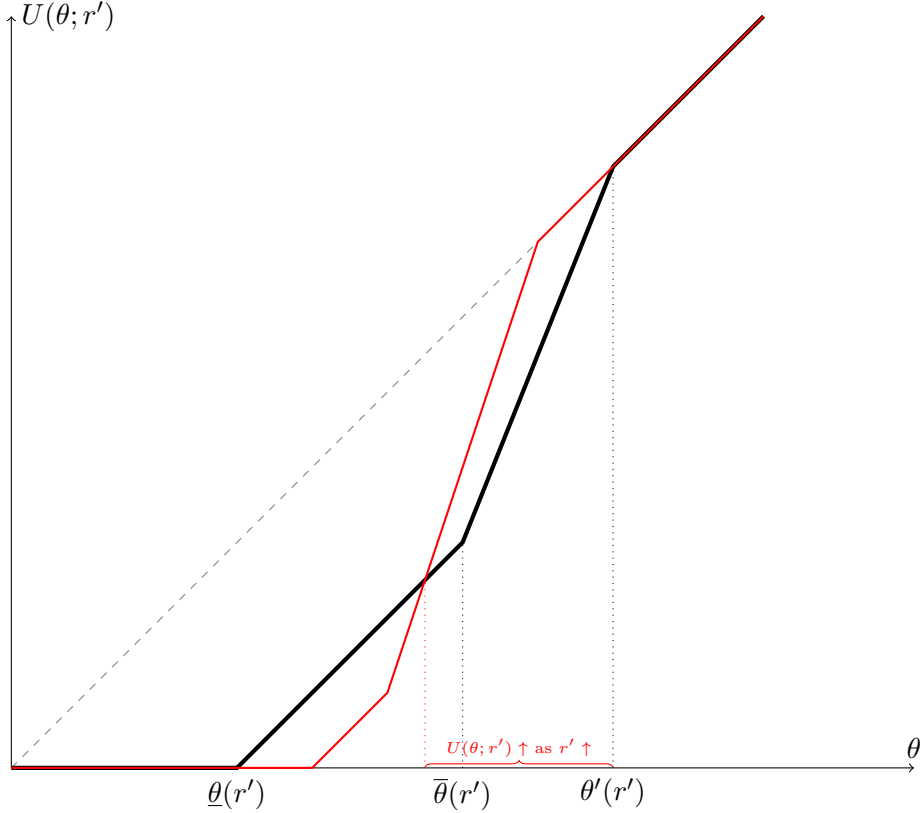
\begin{figure}
\begin{center}
\begin{tikzpicture}[scale=1]
\draw[<->] (0,10) node[right]{{$U(\theta;r')$}} -- (0,0) -- (12,0) node[above]{{$\theta$}};
\draw[dashed,gray] (0,0) -- (10,10);
\draw[dotted] (6,0) node[below]{{$\overline{\theta}(r')$}} -- (6,3);
\draw[dotted] (8,0) node[below]{{$\theta'(r')$}} -- (8,8);
\draw[ultra thick] (0,0) -- (3,0) node[below]{{$\underline{\theta}(r')$}} -- (6,3) -- (8,8) -- (10,10);
\draw[red,dotted] (5.5,0) -- (5.5,2.5);
\draw [red, decorate, decoration={brace}] (5.5,.01) -- node[above] {\tiny{$U(\theta;r')\uparrow$ as $r'\uparrow$}} (8,.01);
\draw[thick, red] (0,0) -- (4,0) -- (5,1) -- (7,7) -- (10,10);
\end{tikzpicture}
\end{center}
\caption{Policymaker's ex post payoff $U(\theta;r')$ as a function of $\theta$, illustrating Proposition~\ref{prop:central}. The \textbf{black solid line} shows $U(\theta;r')$ for a baseline policy $r'$, and the \textbf{red dashed line} shows the payoff under a higher policy $r'' > r'$. For $\sigma > \frac{\,1-\underline{r}\,}{\,2\,\underline{r}\,}$, the red curve lies below the black curve at low $\theta$ but crosses above it at high $\theta$, indicating that aggressive policy hurts the policymaker when fundamentals are weak but helps when fundamentals are strong.}
\label{fig:welfare}
\end{figure}

This result is central to the paper. It shows that when agents' signals are sufficiently noisy, a more aggressive policy ($r'$) harms the policymaker if fundamentals are weak (because intervention becomes costlier) but benefits the policymaker if fundamentals are strong (since a higher $r'$ deters attacks). Conversely, when agents' signals are precise ($\sigma$ is small), raising $r'$ uniformly reduces welfare.

The solid line in Figure \ref{fig:welfare} shows the policymaker's ex post payoff $U(\theta;r')$, given by a piecewise linear function. For $\theta < \underline{\theta}(r')$, the regime is abandoned and the payoff is zero. For $\theta \in [\,\underline{\theta}(r'),\,\overline{\theta}(r')\,]$, the policymaker intervenes and earns a positive payoff. For $\theta \in [\,\overline{\theta}(r'),\,\theta'(r')\,]$, the policymaker faces a positive attack without intervening; and for $\theta > \theta'(r')$, no attack occurs, so the payoff equals $\theta$. The red line corresponds to the payoff under an equilibrium with a higher $r'' > r'$. Increasing $r'$ raises the lower threshold $\underline{\theta}(r')$ but reduces $\overline{\theta}(r')$ and $\theta'(r')$. Those types for which the red line lies above the black line benefit from this increase in $r'$.

Intuitively, increasing $r'$ has a direct negative effect on the types who intervene ($\theta \in [\,\underline{\theta}(r'),\,\overline{\theta}(r')\,]$), as they must pay a higher cost $c(r')$. Yet a larger $r'$ also affects the aggregate attack under $r=\underline{r}$: when $\sigma$ is large, non-intervention makes agents \emph{less aggressive}, reducing $A(\theta;\underline{r},r')$ and thus benefiting strong types ($\theta \in (\overline{\theta}(r'),\,\theta'(r'))$).

A practical way to view these findings is to consider how a policymaker might misjudge the impact of raising $r'$ if they assume agents are naive. In a naive setting, a higher cost of attacking always seems beneficial, since it directly lowers the incentive to mount an attack. Yet when agents are sophisticated enough to infer fundamentals from the observed policy, the very act of setting a high $r'$ can reveal underlying weaknesses if those fundamentals are marginal. This “signal of fear” can invite more aggressive attacks when $\theta$ is not sufficiently strong, countering the policymaker’s goal.

In real-world scenarios, this phenomenon has been documented by \citet{Lorentzen2013}: heightened public security measures in authoritarian contexts can sometimes reveal a leader’s vulnerability, fueling rather than quelling dissent. Conversely, when fundamentals are indeed robust, a higher $r'$ reliably deters attacks by signaling genuine regime strength. Hence, the divergence between naive-agent expectations (“higher $r'$ always helps”) and rational-agent outcomes (“a high $r'$ can backfire if fundamentals are weak”) underscores the need for careful consideration of how observers draw inferences from policy choices.

\emph{Intuition with naive vs. sophisticated agents.} To build further intuition, consider how the impact of an aggressive policy $r'$ differs if agents are \textit{naive} rather than sophisticated. In a naive scenario (agents do not infer fundamentals from policy signals), raising $r'$ \textbf{always} seems beneficial to the policymaker — it directly increases the cost of attacking, which mechanically deters attacks. However, when agents are sophisticated and draw inferences from the observed $r'$, an aggressive intervention can send a \emph{negative signal} about fundamentals. In particular, setting a very high $r'$ when $\theta$ is only moderate may be interpreted as a sign of fear or underlying weakness. This “signal of fear” can invite more aggressive attacks if $\theta$ is not sufficiently high. On the other hand, if $\sigma$ is large (so agents’ private information is very noisy), a strong policy move can reassure agents when fundamentals are in fact high, thereby reducing attacks and benefiting a strong policymaker. In summary, naive agents focus only on the direct deterrence effect of policy (higher $r'$ always discourages attacks), whereas sophisticated agents also consider the informational content of policy — which can either quell or exacerbate attacks depending on the true state and the noise in their private signals.

\section{\label{sec:literature}Literature review}

We now situate our findings within the related literature.

The global-games approach, introduced by \citet{CarlssonvanDamme1993}, revolutionized equilibrium selection in simple $2\times 2$ coordination games by incorporating small, uncorrelated noise into players’ information, thereby yielding a unique equilibrium. Building on this framework, \citet{MorrisShin1998} applied global games to currency attacks, refining the complete-information analysis of \citet{Obstfeld1986}. In a key development, \citet{Angeletosetal2006} showed that policy signaling can reintroduce equilibrium multiplicity\footnote{They also demonstrate that this multiplicity persists under refined out-of-equilibrium beliefs in the spirit of \citet{ChoKreps1987}.} while still assuming that the policymaker is fully informed. In contrast, this paper examines the welfare implications of \emph{active policy} equilibria and explores situations where the policymaker may benefit from remaining uninformed.

Although our analysis focuses on a specific parametric example, it fits within the broader literature on global games with a large player. For instance, \citet{Corsettietal2004} study how a major speculator (e.g., George Soros) can trigger a currency attack that influences smaller speculators, while \citet{Edmond2013} examines a political leader who manipulates information about the regime’s strength, which citizens then rationally discount. In line with \citet{Angeletosetal2006}, our model admits an equilibrium where intervention occurs at intermediate fundamentals.

Research on \emph{multi-stage} revolutionary movements has further enriched our understanding of coordination under uncertainty. Early works by \citet{Kuran1989,Kuran1991} and \citet{Lohmann1994} explore how incomplete information, preference falsification, and cascading behavior shape the timing and likelihood of revolutions. More recent studies complement these insights: \citet{CasperTyson2014} analyze how localized protests can shift elite coordination in a coup, while \citet{Little2017} demonstrates how regimes deploy propaganda to deter collective action—paralleling the signaling rationale in the present framework.

Moreover, \citet{deMesquita2010} develops a formal model in which regime transitions hinge on how individual incentives align with collective revolutionary action, emphasizing the role of private information and threshold-based coordination. In our setting, we focus on how a policymaker’s uninformed status influences these thresholds, revealing that incomplete information can sometimes bolster the regime by mitigating adverse signaling. Complementary studies by \citet{TysonSmith2018} and \citet{Sudduth2017} highlight how strategic behavior and information asymmetries shape outcomes in mass mobilization and coup-risk scenarios.

Further contributions examine the effects of information precision and signaling on coordination. \citet{IachanNenov2015} investigate how the quality of private signals affects the probability of attacks and welfare outcomes in a static global game, while \citet{KyriazisLou2023} propose a signaling game in which a leader moves first and agents coordinate afterward, showing that noise can yield either unique or multiple rationalizable outcomes. In addition, \citet{AngeletosPavan2013} offer selection-free predictions in a generalized global-game model, illustrating how endogenous information structures determine equilibrium outcomes. Relative to these works, this paper focuses on the \emph{ex post} welfare of a policymaker whose information level shapes the signaling process and ultimately the equilibrium.

By allowing the policymaker to remain \emph{uninformed}, this paper departs from the standard assumption in \citet{Angeletosetal2006} and demonstrates that ignorance can sometimes shield the policymaker from the detrimental effects of signaling. This perspective enriches the literature on information structures, policy instruments, and equilibrium selection.

\section{\label{sec:conclusion}Conclusion}

This paper presents a simplified regime-change game in which a policymaker can preemptively undertake a costly intervention to deter a coordinated attack. Our analysis reveals that, when agents' private signals are imprecise, aggressive signaling through policy intervention has a double-edged effect: it benefits a strong policymaker by deterring attacks, while it harms a weak policymaker by imposing high intervention costs. In essence, the policymaker's detailed knowledge of the fundamentals may backfire if it intensifies the signaling channel and prompts adverse coordination among agents.

Although these results are derived from a stylized parametric example, they raise important questions about the optimal degree of knowledge a policymaker should have before initiating preventive actions.

\paragraph{Acknowledgments.} I am grateful to Christian Hellwig, Ilina Ilyasova, Konstantin Kozlov, the editor and two anonymous referees for their insightful comments and suggestions.

Any remaining errors are mine.

\paragraph{Conflict of Interest.}
The author declares that he has no known competing financial interests or personal relationships that could have appeared to influence the work reported in this paper.

\appendix
\section{Proofs}\label{approofs}

\subsection{Proof of Proposition \ref{prop:central}}\label{app:prop2}
\begin{proof}
The proof examines the policymaker's strategic response to potential attacks as a function of $\theta$ and $r'$. When $r = \underline{r}$, the aggregate attack is 
\begin{equation*}
A(\theta;\underline{r},r') = \Pr\{x \le x^* \mid \theta\} = 
\begin{dcases}
1, & \text{if } \theta < 2\sigma(\underline{\theta} - 1) + \overline{\theta},\\[1mm]
\underline{\theta} + \frac{\overline{\theta} - \theta}{2\sigma}, & \text{if } (2\sigma - 1)\underline{\theta} + \overline{\theta} \le \theta < 2\sigma\,\underline{\theta} + \overline{\theta},\\[1mm]
0, & \text{if } \theta \ge 2\sigma\,\underline{\theta} + \overline{\theta}\,.
\end{dcases}
\end{equation*}
Substituting the expression for $\overline{\theta}$ yields 
\begin{equation}
\label{eq:aggatt}
A(\theta;\underline{r},r') = 
\begin{dcases}
1, & \text{if } \theta < \Big[1 - 2\sigma\,\frac{\underline{r}}{\,1-\underline{r}\,}\Big]\underline{\theta},\\[1mm]
\frac{2\sigma + \Big[1 - 2\sigma\,\frac{\underline{r}}{\,1-\underline{r}\,}\Big]\underline{\theta} - \theta}{2\sigma}, & \text{if } \Big[1 - 2\sigma\,\frac{\underline{r}}{\,1-\underline{r}\,}\Big]\underline{\theta} \le \theta < 2\sigma + \Big[1 - 2\sigma\,\frac{\underline{r}}{\,1-\underline{r}\,}\Big]\underline{\theta},\\[1mm]
0, & \text{if } \theta \ge 2\sigma + \Big[1 - 2\sigma\,\frac{\underline{r}}{\,1-\underline{r}\,}\Big]\underline{\theta}\,. 
\end{dcases}
\end{equation}

For any given $r'\in(\underline{r},\tilde{r}]$, the signaling equilibrium is determined and thus the policymaker’s ex post welfare as a function of $\theta$ can be written piecewise.  Using $C(r')=\underline{\theta}(r')$ from \eqref{eq:fortlow}, we obtain
\begin{equation}
\label{eq:pmwelf}
U(\theta; r') =
\begin{dcases}
0, & \text{if } \theta < \underline{\theta}(r'),\\
\theta - \underline{\theta}(r'), & \text{if } \underline{\theta}(r') \le \theta < \overline{\theta}(r'),\\
\bigl(1 + \tfrac{1}{2\sigma}\bigr)\theta - \Bigl(\tfrac{1}{2\sigma}-\tfrac{\underline{r}}{1-\underline{r}}\Bigr)\underline{\theta}(r') - 1, & \text{if } \overline{\theta}(r') \le \theta < \theta'(r'),\\
\theta, & \text{if } \theta \ge \theta'(r').
\end{dcases}
\end{equation}
Here $\overline{\theta}(r')$ and $\theta'(r')$ are defined as in (7) above and depend on $r'$.  (See \eqref{eq:pmnewwlfr} for an alternative representation in terms of $\underline{\theta}$.) The welfare formula \eqref{eq:pmwelf} makes explicit how the policymaker’s payoff in each regime depends on $\theta$ and $r'$.

Observe that for $r' > \underline{r}$, $\underline{\theta}(r')$ is increasing in $r'$. Since $\underline{\theta}(r')$ is given by \eqref{eq:fortlow}, we have 
\begin{equation*}
\frac{\partial \underline{\theta}(r')}{\partial r'} = (r' - \underline{r}) > 0\,.
\end{equation*}

We can rewrite the policymaker's ex post welfare \eqref{eq:pmwelf} as\footnote{Note that the cutoff levels $\overline{\theta}$ and $\theta'$ also depend on $r'$ (equivalently on $\underline{\theta}$).} 
\begin{equation}
\label{eq:pmnewwlfr}
U(\theta; \underline{\theta}) = 
\begin{dcases}
0, & \text{if } \theta < \underline{\theta},\\[1mm]
\theta - \underline{\theta}, & \text{if } \underline{\theta} \le \theta < \overline{\theta},\\[1mm]
\Big(1 + \frac{1}{2\sigma}\Big)\theta - \Big(\frac{1}{2\sigma} - \frac{\underline{r}}{\,1-\underline{r}\,}\Big)\underline{\theta} - 1, & \text{if } \overline{\theta} \le \theta < \theta',\\[1mm]
\theta, & \text{if } \theta \ge \theta'\,,
\end{dcases}
\end{equation}
and therefore the sign of $\frac{\partial U(\theta;r')}{\partial r'}$ is the same as the sign of $\frac{\partial U(\theta; \underline{\theta})}{\partial \underline{\theta}}$.

From \eqref{eq:pmnewwlfr}, it is immediate that $\frac{\partial U(\theta; \underline{\theta})}{\partial \underline{\theta}} = -1 < 0$ for $\theta \in [\,\underline{\theta},\,\overline{\theta}\,]$.

For $\theta \in (\overline{\theta}, \theta')$, we have 
\begin{equation*}
\frac{\partial U(\theta; \underline{\theta})}{\partial \underline{\theta}} = -\Big(\frac{1}{2\sigma} - \frac{\underline{r}}{\,1-\underline{r}\,}\Big) \lesseqgtr 0 \quad \text{as} \quad \sigma \lesseqgtr \frac{\,1-\underline{r}\,}{\,2\,\underline{r}\,}\,,
\end{equation*}
implying that if $\sigma > \frac{\,1-\underline{r}\,}{\,2\,\underline{r}\,}$, then $U(\theta; \underline{\theta})$ is increasing in $\underline{\theta}$ — and thus $U(\theta; r')$ is increasing in $r'$. 
\end{proof}

\bibliographystyle{apalike}
\bibliography{WelfareCoord}

\end{document}